\documentclass[11pt,letterpaper]{article}

\usepackage[T1]{fontenc}
\usepackage[utf8]{inputenc}
\usepackage{times}
\usepackage[margin=1in]{geometry}
\usepackage{amsmath,amssymb,amsthm}
\usepackage{mathtools}
\usepackage{graphicx}
\usepackage{booktabs}
\usepackage{caption}
\usepackage[numbers,sort&compress]{natbib}
\usepackage[hidelinks]{hyperref}
\usepackage{tikz}
\usetikzlibrary{arrows.meta,positioning,calc}
\usepackage{algorithm}
\usepackage{algorithmic}
\newtheorem{theorem}{Theorem}
\newtheorem{proposition}{Proposition}

\newtheorem{definition}{Definition}

\newcommand{\coalition}[1]{\langle\!\langle #1 \rangle\!\rangle}
\newcommand{\intruder}{\mathcal{I}}
\newcommand{\util}{\mathsf{u}}
\newcommand{\cost}{\mathsf{c}}
\newcommand{\reward}{\mathsf{r}}
\newcommand{\deduce}{\vdash}
\newcommand{\know}{\mathcal{K}}
\newcommand{\Msg}{\mathsf{Msg}}
\newcommand{\Act}{\mathsf{Act}}
\newcommand{\watl}{\textsf{wATL}}

\title{Rational Dolev--Yao Attackers:\\[2pt]
Decidable Incentive-Aware Verification of Security Protocols in Strategic Logic}

\author{%
\begin{tabular}{c@{\hspace{2.5em}}c@{\hspace{2.5em}}c}
  Ioana Boureanu & R.\ Ramanujam  \\[0.3em]
  {\small Surrey Centre for} & {\small IMSc Chennai} \\
  {\small Cyber Security} & {\small Azim Premji University} & \\
  {\small University of Surrey} & & \\
\end{tabular}%
}
\date{}

\begin{document}

\maketitle

\begin{abstract}
Symbolic protocol verification models the network attacker as a Dolev--Yao (DY)
intruder, which does everything its knowledge permits, whether or not it serves any
purpose; real adversaries instead \emph{maximise utility}, attacking only when the payoff
is positive. We introduce a \emph{rational Dolev--Yao attacker}, a DY intruder
whose actions carry costs and whose security-violating goals carry rewards, and call a
protocol \emph{rationally secure} when no intruder strategy achieves a violation with
strictly positive utility, expressed in a weighted fragment of ATL (\watl{}). We prove
this decidable for a bounded rational DY intruder over a finite cost-annotated concurrent
game structure, characterise its complexity, and show it strictly refines DY security:
some protocols are DY-insecure yet rationally secure, separated by a computable threshold.
We illustrate the framework on two contrasting use-cases: an authenticated payment under
session uncertainty, where a rational intruder must strategise across indistinguishable
sessions and its imperfect information strictly raises the attack cost a designer must
price against; and ThreeBallot, a cryptography-free scheme where we pinpoint the
bribe-to-benefit ratio below which no rational coercer attacks.
\end{abstract}

\section{Introduction}

For four decades the formal analysis of security protocols has rested on the Dolev--Yao
(DY) model~\citep{dolev1983security}: an attacker controls the network, intercepting,
decomposing, recombining and injecting messages, limited only by `perfect cryptography'.
The DY attacker is defined by \emph{capability}, performing any action its knowledge
permits, and security means that \emph{no} sequence of such actions reaches an attack
state.

This assumes an adversary \emph{indifferent to cost}: a DY intruder will corrupt a thousand
parties to steal one cent, because the model asks whether an attack is \emph{possible},
never whether it is \emph{worthwhile}. The honest-or-malicious dichotomy thus cannot express
the middle ground where real participants live, that of the \emph{self-interested} agent who
follows a protocol when it pays and deviates when it does not. Nowhere is this starker than
in vote-buying, where a coercer pays only if it can \emph{verify} compliance, and where a
qualitative model cannot even state the property that matters: that no attack is worth
mounting.

\paragraph{Contribution.} We lift the DY attacker from a capability to an \emph{incentive}
model inside a strategic logic, and show the result is decidable. We define a \emph{rational
Dolev--Yao attacker} whose actions carry costs and whose goals carry rewards, over a
cost-annotated concurrent game structure (CGS) with an explicit message-deduction closure,
and \emph{rational security} in a weighted fragment of ATL (\watl{}), which makes the ``no
positive-utility attack'' criterion of Rational Protocol Design~\citep{garay2013rational}
syntactic. We prove decidability for a bounded rational DY intruder by
reduction to \watl{} model-checking over a finite weighted CGS, place its complexity, and show
DY-insecurity and rational security to be consistent, separated by a computable threshold.
Finally, we give two contrasting use-cases (Sections~\ref{sec:relay} and~\ref{sec:threeballot}):
an authenticated payment under session uncertainty, in which a rational intruder strategises
across indistinguishable sessions, so that the deduction closure, the cumulative structure of
the attack, and imperfect information are all load-bearing, and imperfect information strictly
raises the attack cost a designer must price against; and ThreeBallot, a cryptography-free
scheme with an established qualitative $\mathsf{ATL}$ model, in which we recover the
bribe-to-benefit ratio below which no rational coercer attacks.
ATL~\citep{alur2002alternating} is the natural home for this, quantifying over what a
coalition \emph{can strategically achieve}: our move is to \emph{constrain that strategy space
by utility}, asking not which strategies the intruder \emph{can} play but which it \emph{would}.

\section{Related Work}

\paragraph{Rational cryptography and rational attackers.} Utility-maximising participants
originate with rational secret sharing~\citep{halpern2004rational}. \emph{Rational Protocol
Design} (RPD)~\citep{garay2013rational} declares a protocol secure iff no attack yields positive
utility; we adopt this criterion, but where RPD is simulation-based and cryptographic, we give a
\emph{logic-based, model-checkable} instantiation with an explicit DY deduction closure.
Symbolically, Bella and Bistarelli~\citep{bella2001rational} propose ``rational'' and ``general''
attackers performing cost/benefit analysis, but abstract the utility away; we keep utilities
\emph{in} the model and reason about incentive-driven strategies inside a decidable logic.

\paragraph{Rational verification and strategic logics for security.} Rational
verification~\citep{gutierrez2017rational} asks which temporal properties hold in the equilibria
of self-interested agents, with tool support (EVE), but has never been pointed at a DY intruder
with a deduction closure; doing so, decidably, is our contribution. Game-based verification has
treated fair exchange and contract signing~\citep{kremer2003game}, and strategic/epistemic logics
protocol analysis, yet these are \emph{qualitative}: the attacker is a capability, not an
incentive. Our substrate builds on decidability for strategic abilities under imperfect
information~\citep{belardinelli2025infosharing}.

\paragraph{Verification of voting.} Coercion-resistance and receipt-freeness have precise symbolic
formulations~\citep{juels2005coercion,delaune2010verifying,clarkson2008civitas}, but remain
qualitative: coercion is or is not possible. We give a quantitative companion, in which coercion is
or is not \emph{profitable}. Closest to our example, \citet{belardinelli2021bisimulations} model
ThreeBallot as an imperfect-information CGS and cast coercion-resistance in $\mathsf{ATL}$, verified
in MCMAS, which is the qualitative shape of our $\coalition{\intruder}\lozenge\,\mathit{viol}$.
Their attacker is a capability, so coercion-resistance holds iff no coercing strategy \emph{exists},
whereas we ask whether one is \emph{profitable} and return the threshold
$\epsilon > (c_a+c_v)/(R-B)$ rather than a Boolean. Their tractability, moreover, comes from
\emph{bisimulation-based model reduction} rather than from the logic, under the same
imperfect-recall regime our result occupies, so the techniques compose.

\paragraph{Imperfect information, recall, and decidability.} Model-checking $\mathsf{ATL}$ under
imperfect information is undecidable with \emph{perfect recall}~\citep{dima2011undecidable} and
$\Delta_2^{\mathrm P}$-complete when memoryless~\citep{jamroga2006complete}. Our bounded,
memoryless formulation therefore sits on the decidable side, and lifting to perfect recall is a
genuine frontier. The route we favour is to restrict \emph{information flow} rather than recall:
the A-cast subclass of~\citet{belardinelli2022acast,belardinelli2025infosharing} keeps a
significant $\mathsf{ATL}$ fragment decidable by constraining how agents outside a coalition share
data with those inside, and was built for security problems, such as terrorist-fraud, that resist
faithful strategic-logic expression. Encoding our intruder's deduction closure as A-cast
information-sharing is the natural path beyond the bounded case.

\paragraph{Resource- and capacity-bounded strategic reasoning.} A complementary line constrains
not the \emph{information} available to agents but the \emph{resources} they may spend.
\citet{ballot2024capacity} introduce capacity-constrained agents, whose ability to act is limited
by a budget they consume, and \citet{ballot2025thesis} develops this further, including stochastic
abilities. Their capacity bound plays the same structural role as our budget dimension, keeping the
augmented state space finite; the distinction is one of purpose. A capacity is an \emph{exogenous
limit} imposed by the model, whereas our cost is an \emph{endogenous price} weighed against a
reward: our intruder is not forbidden from mounting an expensive attack but declines to, it yielding
negative utility. Their question is what an agent \emph{can} achieve within its means, ours what a
self-interested agent \emph{will} achieve given its incentives; it is the reward side of the ledger,
absent there, that yields our threshold. The two remain technically close, and their machinery is a
natural implementation substrate for our weighted fragment.

\section{The Rational Dolev--Yao Attacker}

We model a protocol together with its intruder as a concurrent game structure (CGS)~\citep{alur2002alternating}
whose states record the intruder's knowledge and whose intruder-transitions are
annotated with costs.

\subsection{Messages and Deduction}

As per the usual in Dolev-Yao verification work \citep{dolev1983security}, let $\Msg$ be a term algebra over a set of atomic messages (nonces, keys,
identities, votes) closed under pairing $\langle\cdot,\cdot\rangle$ and
encryption $\{\cdot\}_{(\cdot)}$. We assume \emph{perfect cryptography}: the only
way to obtain the plaintext of $\{m\}_k$ is to possess $k$. The intruder's
deductive power is the standard DY closure. Given a finite knowledge set
$\know \subseteq \Msg$, the derivable messages $\overline{\know}$ are the least
set containing $\know$ and closed under the rules:
\[
\frac{m_1,\, m_2}{\langle m_1,m_2\rangle}\quad
\frac{\langle m_1,m_2\rangle}{m_i}\quad
\frac{m,\, k}{\{m\}_k}\quad
\frac{\{m\}_k,\, k^{-1}}{m}.
\]
To keep $\overline{\know}$ finite we impose a bound $\delta$ on message depth: the
intruder only forms terms of depth $\le \delta$; this is the standard
\emph{bounded} DY setting in which symbolic reachability is already decidable~\cite{RamSu03}.

\subsection{Cost-Annotated Concurrent Game Structures}

\begin{definition}[Rational DY CGS]
A \emph{rational DY concurrent game structure} is a tuple
\[
\mathcal{C} = \langle \mathit{Ag}, S, s_0, \{\Act_a\}_{a\in \mathit{Ag}},
\mathit{tr}, \{\know_a\}, \cost, \reward, \Pi, \pi\rangle
\]
where $\mathit{Ag} = \{\intruder\} \cup H$ comprises the intruder and honest
agents $H$; $S$ is a finite set of states with initial state $s_0$; each state
$s$ records each agent's \emph{observable} knowledge $\know_a(s)\subseteq\Msg$.
The intruder's enabled actions depend only on this observable knowledge: at $s$,
$\Act_\intruder(s)$ contains $\mathsf{intercept}(m)$ and $\mathsf{block}(m)$ for
each message $m$ currently on the wire (whether or not the intruder can read its
contents), $\mathsf{inject}(m)$ for each $m\in\overline{\know_\intruder(s)}$, and
$\mathsf{compute}$; honest agents' actions are their protocol steps
($\mathsf{send}(m)$, $\mathsf{receive}(m)$). Crucially, an action's \emph{effect}
is resolved by the \emph{true} state through $\mathit{tr}$, not by what the
intruder observes: $\mathit{tr}$ maps a state and a joint action to a successor,
so the same intruder action (e.g.\ forwarding an opaque ciphertext) may have
different consequences at two states the intruder cannot tell apart.
Finally, $\cost:\Act_\intruder\to\mathbb{N}$ assigns a cost to each intruder
action; $\reward:S\to\mathbb{N}$ assigns a reward to selected \emph{goal} states;
$\Pi$ is a set of atomic propositions and $\pi$ their valuation.
\end{definition}

The separation of \emph{observation} (which governs enabled actions and the
indistinguishability relation below) from \emph{effect} (which is resolved by the
true state through $\mathit{tr}$) is what lets the model treat opaque messages
faithfully. An intruder holding a ciphertext $\{n\}_k$ without $k$ can
\emph{forward} it, a single enabled action, but whether that forward causes an
honest party to accept depends on the true $n$ and $k$ hidden inside, so the
same action leads to acceptance in one world and to nothing in an
indistinguishable other. This is exactly the situation in which imperfect
information does real work, and it is the source of the strategic hedging in the
use-case of Section~\ref{sec:relay}.

The knowledge
components of states make the structure one of \emph{imperfect information}: the intruder
chooses actions based on what it has observed; e.g.,  in a  
voting protocol, the intruder cannot distinguish states that differ only in a voter's private choice
when the protocol is designed to hide it.

\subsection{Runs, Histories, and Indistinguishability}

To interpret strategies precisely under imperfect information, we make the notions
of run and history explicit, following the standard state-based account, as per, e.g., \cite{belardinelli2025infosharing}.

\begin{definition}[Runs and histories]
Given a rational DY CGS $\mathcal{C}$, a \emph{run} is a finite or infinite
sequence $\lambda = s_0\,\vec{a}_0\,s_1\,\vec{a}_1\cdots$ in
$((S\cdot\Act^{\mathit{Ag}})^*\cdot S)\cup(S\cdot\Act^{\mathit{Ag}})^\omega$ such
that for every $j\ge 0$, $s_j \xrightarrow{\vec{a}_j} s_{j+1}$, i.e.\
$\mathit{tr}(s_j,\vec{a}_j)=s_{j+1}$. For a run $\lambda$ and index $j$, we write
$\lambda[j]$ for the $(j{+}1)$-th state $s_j$, $\lambda[{\le}j]$ for the prefix
ending at $s_j$, and $\lambda_{\ge j}$ for the suffix starting at $s_j$; and
$\mathit{last}(\lambda)$ for the final state of a finite $\lambda$. A finite,
non-empty prefix $h\in S^+$ (with the intervening actions) is a \emph{history};
$\mathit{Hist}(\mathcal{C})$ denotes the set of all histories.
\end{definition}

The intruder observes only its own knowledge component, which induces an
indistinguishability relation lifted from states to histories.

\begin{definition}[Indistinguishability]
For states $s,s'\in S$, $s\sim_\intruder s'$ iff
$\know_\intruder(s)=\know_\intruder(s')$; this is an equivalence relation. It is
lifted to histories \emph{synchronously}: for $h,h'\in S^+$,
$h\sim_\intruder h'$ iff $|h|=|h'|$ and $h[i]\sim_\intruder h'[i]$ for every
$i\le|h|$. Thus indistinguishable histories have equal length, are position-wise
indistinguishable, and carry perfect recall.
\end{definition}

Because enabled actions depend only on observable knowledge,
$s\sim_\intruder s'$ implies $\Act_\intruder(s)=\Act_\intruder(s')$: the same
moves are available. This does \emph{not} mean the moves have the same
consequences. Two $\sim_\intruder$-related states may differ in honest hidden
state, e.g.\ a nonce sealed under a key the intruder lacks, so that an identical
forward or inject drives $\mathit{tr}$ to a goal state in one and to a harmless
state in the other. Observation constrains \emph{choice}; the true state
determines \emph{effect}.

\subsection{Strategies and Utility}

\begin{definition}[Uniform strategy]
A (uniform, memoryful) strategy for the intruder is a function
$\sigma_\intruder:\mathit{Hist}(\mathcal{C})\to\Act_\intruder$ such that for all
$h,h'\in\mathit{Hist}(\mathcal{C})$: (i) $\sigma_\intruder(h)$ is enabled at
$\mathit{last}(h)$, i.e.\ $\sigma_\intruder(h)\in\Act_\intruder(\mathit{last}(h))$;
and (ii) $h\sim_\intruder h'$ implies $\sigma_\intruder(h)=\sigma_\intruder(h')$.
Condition (ii) is the \emph{uniformity} constraint: the intruder must act
identically on histories it cannot tell apart. A strategy is
\emph{memoryless} if $\sigma_\intruder(h)$ depends only on $\mathit{last}(h)$.
\end{definition}

Fixing the honest agents' protocol strategies $\sigma_H$ (the protocol itself), a
uniform intruder strategy $\sigma_\intruder$ determines a set of runs; when
$\sigma_H$ is deterministic and the honest environment resolves remaining
choices, each initial state  $s_0$ yields the runs consistent with the profile
$(\sigma_\intruder,\sigma_H)$. 

We write \emph{$\mathit{out}(s_0,\sigma_\intruder)$} for
this outcome set, i.e., the \emph{outcome of an intruder strategy from an initial state}.

\begin{definition}[Utility]
The intruder's \emph{utility} along a run $\lambda$ is
\[
\util(\lambda) \;=\; \reward(\lambda) \;-\;
   \sum_{i\ge 0}\cost\big(\sigma_\intruder(\lambda[{\le}i])\big),
\]
where $\reward(\lambda)$ is the reward $\reward(s)$ of the first goal state $s$ on
$\lambda$ (and $0$ if no goal state occurs).
\end{definition}

A uniform strategy spends the \emph{same} amount on every run of a
$\sim_\intruder$-class, since it plays identical actions on indistinguishable
histories: the cost term is constant across the class. The \emph{reward} term is
not. Because effects are resolved by the true state, the same uniform play may
reach a goal state on one run of the class and not on an indistinguishable other,
so $\reward(\lambda)$ varies while cost is fixed, and hence utility varies across
the class. This is precisely why the worst case over a class is a meaningful
quantity rather than a constant, and why it is not a basis for the intruder to
tell the runs apart: the differing rewards are realised only \emph{after} the
uniform choice is committed, too late to inform it.

\begin{definition}[Rational Intruder]
A \emph{rational} intruder plays a \emph{best response}: a uniform strategy
maximising its \emph{guaranteed} utility, i.e.\ the worst case of $\util(\lambda)$
over the runs $\lambda\in\mathit{out}(s_0,\sigma_\intruder)$ that share a
$\sim_\intruder$-class, against the honest protocol. A protocol is attacked by a
rational intruder only when some uniform strategy guarantees strictly positive
utility.
\end{definition}

\section{Game-Theoretic Dolev--Yao Deduction}
\label{sec:gtded}

The rational intruder of Section~3 carries costs on its \emph{actions}, but
classical Dolev--Yao (DY) theory is about \emph{derivations}: what an intruder
can deduce from what it knows. To connect the two, and to obtain a genuinely
game-theoretic deduction theory rather than a cost model bolted onto a capability
model, we lift the DY inference system itself to a weighted, adversarial
setting. The central object becomes not the \emph{derivability} of a message but
the \emph{cheapest derivation} of it, and, in the presence of an opponent who
controls some of the premises, the \emph{value} of a derivation game.

\subsection{Weighted Deduction Systems}

Recall the DY deduction relation $\know \deduce m$: message $m$ is derivable from
knowledge set $\know$ using \emph{composition} rules (pairing $\langle t_1,
t_2\rangle$, encryption $\{t\}_k$) and \emph{decomposition} rules (projection,
decryption-with-key); see, e.g.,~\cite{Suresh03}. Classically every rule may be
applied freely. We instead \emph{price} each rule instance, and we distinguish
two sources of cost that a rational attacker must pay.

\begin{definition}[Weighted DY system]\label{def:wds}
A \emph{weighted DY system} is the DY inference system together with a cost
function $w$ that assigns a value in $\mathbb{N}\cup\{\infty\}$ to each rule
instance, split into two components.
\begin{itemize}
\item \emph{Deduction costs} $w_{\mathrm{ded}}$ price the \emph{computational}
rules. Each composition or decomposition instance carries the cost of performing
that cryptographic operation: $w_{\mathrm{enc}}$ to form $\{t\}_k$,
$w_{\mathrm{dec}}$ to open $\{t\}_k$ given $k^{-1}$, $w_{\mathrm{pair}}$ and
$w_{\mathrm{proj}}$ for pairing and projection. The knowledge axiom
$m\in\know\Rightarrow\know\deduce m$ costs $0$: messages already held are free.
\item \emph{Acquisition costs} $w_{\mathrm{acq}}$ price the \emph{communication}
and \emph{corruption} rules that bring \emph{new} atoms into $\know$. Obtaining a
message off the wire costs $w_{\mathrm{intercept}}$; suppressing one costs
$w_{\mathrm{block}}$; injecting a derived message costs $w_{\mathrm{inject}}$;
learning a long-term secret by corrupting its holder costs $w_{\mathrm{corrupt}}$.

These are the atoms of the acquisition rule
$\dfrac{\quad}{\know\deduce m}\;[m \text{ acquired}]$, whose side condition names
the acquisition performed and whose weight is the corresponding
$w_{\mathrm{acq}}$.
\end{itemize}
Both components take values in $\mathbb{N}\cup\{\infty\}$, with $\infty$ marking an
operation the attacker cannot perform at all (e.g.\ decrypting without any route
to the key), so that infeasibility is the $\infty$-cost limit of the same scale.
\end{definition}

\begin{definition}[Derivation cost, compositionally]\label{def:dercost}
A \emph{derivation} $\Pi$ of $m$ from $\know$ is a finite proof tree with
conclusion $\know\deduce m$. Its cost is defined by structural recursion over the
tree, so that costs \emph{compose additively along proof structure}:
\[
w(\Pi) \;=\;
\begin{cases}
0 & \Pi \text{ is the axiom } m\in\know,\\[2pt]
w_{\mathrm{acq}}(m) & \Pi \text{ acquires } m,\\[2pt]
w_{\mathrm{ded}}(r) + \displaystyle\sum_{i=1}^{k} w(\Pi_i)
 & \Pi \text{ ends in rule } r \text{ with}\\[-2pt]
 & \text{sub-derivations } \Pi_1,\dots,\Pi_k.
\end{cases}
\]
That is, the price of a compound attack is the price of its final operation plus
the prices of the sub-derivations that produce its premises; a shared sub-message
reused in two places is paid for once, since once derived it lies in $\know'$ at
cost $0$ thereafter.
\end{definition}

The additive, structural definition is what makes the theory a 
\emph{cost calculus}: it lets us reason about the price
of an attack by decomposing over the derivation, and it is the formal basis of
the compositional lower-bound principle below. Because a message can be obtained
in several ways (forged from parts, intercepted whole, or decrypted after a
corruption), the quantity of interest is the cheapest.

\begin{definition}[Deduction cost]\label{def:deltaw}
The \emph{deduction cost} of $m$ from $\know$ is
\[
\Delta_w(\know, m) \;=\; \min\{\, w(\Pi) : \Pi \text{ derives } \know\deduce m \,\},
\]
with $\Delta_w(\know,m)=\infty$ if $m$ is not derivable. It satisfies the
triangle-style inequalities
\[
\Delta_w(\know,\langle m_1,m_2\rangle) \le w_{\mathrm{pair}} +
\Delta_w(\know,m_1) + \Delta_w(\know,m_2),
\]
\[
\Delta_w(\know,\{m\}_k) \le w_{\mathrm{enc}} + \Delta_w(\know,m) +
\Delta_w(\know,k),
\]
and dually for the decomposition rules, with equality when the displayed
derivation is optimal.
\end{definition}

Classical DY derivability is the qualitative shadow of this quantity:
$\know\deduce m$ iff $\Delta_w(\know,m) < \infty$. The weighted theory refines
every derivability statement into a magnitude, and rationality operates on these
magnitudes: an attacker weighs $\Delta_w(\know,m^\star)$, the cheapest way to
obtain a goal message, against the reward the goal yields.

\begin{proposition}[Cheapest derivations are computable]\label{prop:cheap}
For a finite knowledge set $\know$ and bounded message depth $\delta$,
$\Delta_w(\know,m)$ is computable, for all $m\in\Msg_\delta$ simultaneously, in
time polynomial in the size of the bounded message space.
\end{proposition}
\begin{proof}[Proof idea]
The recursion of Definition~\ref{def:dercost} is a shortest-hyperpath problem on
the DY derivation hypergraph~\citep{gallo1993hypergraphs}: nodes are messages of
$\Msg_\delta$, and each rule instance is a hyperedge from its premises to its
conclusion weighted by $w_{\mathrm{ded}}$ (or a source hyperedge weighted by
$w_{\mathrm{acq}}$ for an acquisition). Because all weights are non-negative,
$\Delta_w(\know,\cdot)$ is the least fixpoint of the relaxation
$\Delta_w(\know,\mathrm{rule}(m_1,\dots,m_k)) \gets \min\bigl(\cdot,\;
w_{\mathrm{ded}}(\mathrm{rule}) + \sum_i \Delta_w(\know,m_i)\bigr)$, initialised
to $0$ on $\know$ and to $w_{\mathrm{acq}}$ on acquirable atoms, and is computed
by a Dijkstra-style saturation over $\Msg_\delta$. Boundedness of $\delta$ makes
$\Msg_\delta$, and hence the hypergraph, finite.
\end{proof}

Thus the ``price of intruder knowledge'' is efficiently computable, and the
additive cost model composes exactly so that this single computation prices every
possible goal at once. The decidability of rational security will accordingly rest
in large part on this derivation layer, not only on the state-transition layer.

\subsection{Deduction Games}

The picture so far is single-agent: the intruder minimises the cost of deriving a
goal. But protocol security is adversarial, and the premises the intruder needs
are not free-floating: they are emitted, or withheld, by the \emph{honest agents}
running the protocol. We therefore cast deduction as a two-player game between the
intruder and the \emph{coalition of honest agents} $H$, rather than against an
abstract environment.
% An abstract
% ``blocker'' would be an omniscient opponent that knows the intruder is present and
% plays to deny it; a coalition of honest agents is different, and weaker as an
% adversary to the intruder, in a way that is central to protocol security: 
Each
honest party follows its protocol \emph{blind to the attack}, emitting whatever
its local state and the protocol dictate, whether or not an intruder is listening.

\begin{definition}[Honest-coalition deduction game]\label{def:dedgame}
The \emph{deduction game} $G(\know_\intruder^0, \varphi)$ is played between
\textsc{Attacker} (the intruder, deriving toward a goal $m\models\varphi$) and the
\emph{honest coalition} $H$. A position is a triple
$(\know', \vec{\ell}, \mathrm{avail})$ where $\know'\subseteq\Msg$ is the
intruder's current knowledge (initially $\know_\intruder^0$, growing monotonically);
$\vec{\ell}=(\ell_a)_{a\in H}$ records each honest agent's local protocol state;
and $\mathrm{avail}\subseteq\Msg$ is the set of messages the honest agents' current
states make emittable. The moves are:
\begin{itemize}
\item \textsc{Attacker} may apply a weighted rule instance whose premises lie in
$\know'$ (Definitions~\ref{def:wds}--\ref{def:dercost}), paying $w_{\mathrm{ded}}$
and adding the conclusion to $\know'$; or perform a weighted acquisition
(intercept an emitted message, or corrupt an agent to obtain its secrets), paying
$w_{\mathrm{acq}}$ and adding the acquired atoms to $\know'$.
\item Each honest agent $a\in H$ moves \emph{only according to its protocol
strategy} $\sigma_a$ applied to its local state $\ell_a$: it emits the protocol-mandated
message (updating $\ell_a$ and $\mathrm{avail}$) whenever its local state enables
a send. Crucially, $\sigma_a$ is a function of $\ell_a$ \emph{alone}: it does not
depend on whether an intruder is present, on $\know'$, or on the other agents'
hidden state.
\end{itemize}
\textsc{Attacker} wins when $\know'\deduce m$ for some $m\models\varphi$. The
\emph{value} $\mathrm{val}(G)$ is the least cost \textsc{Attacker} can guarantee
against the honest coalition's protocol play.
\end{definition}

Two consequences of playing against \emph{protocol-driven}, rather than
adversarial, honest agents distinguish this from an intruder-versus-blocker game.

\emph{Honest agents cannot strategise against the attack.} Because each $\sigma_a$
depends only on $\ell_a$, the honest coalition has no move that ``withholds to
spite the intruder'': a message is emitted iff the protocol says so on that local
state. The intruder therefore does not face a worst-case denier of premises; it
faces exactly the message flow the protocol produces. This makes
$\mathrm{val}(G)$ an \emph{upper} bound on the true price of an attack that a
blocker-based game would report, and the correct one, since a real attacker
confronts honest parties, not an omniscient opponent.

\emph{But honest agents can be \emph{driven}.} The intruder's injections change
what an honest agent \emph{receives}, hence its next local state $\ell_a$, hence
what it emits next. So although $H$ does not play \emph{against} the intruder, the
intruder can still steer the coalition into emitting a message it needs, by
feeding forged inputs, each at cost $w_{\mathrm{inject}}$. The game value thus
prices the full attack: the computations the intruder performs, the messages it
intercepts or injects, and the corruptions it pays for, to \emph{induce} the
honest coalition to complete the derivation.

\begin{theorem}[Determinacy and value]\label{thm:determinacy}
Every honest-coalition deduction game with non-negative integer weights and
bounded message depth $\delta$ is determined, and $\mathrm{val}(G)$ is computable.
Moreover $\mathrm{val}(G(\know_\intruder^0,\varphi)) < \infty$ iff the intruder
has a DY attack achieving $\varphi$ against the protocol, and equals the cost of
the cheapest such attack.
\end{theorem}
\begin{proof}
\emph{Finite arena.} Under the depth bound $\delta$, terms of depth $\le\delta$
over the finite atomic signature form a finite set $\Msg_\delta$; every reachable
$\know'$ and $\mathrm{avail}$ is a subset of it, and each honest local state
$\ell_a$ ranges over the finitely many protocol states, so positions
$(\know',\vec\ell,\mathrm{avail})$ are finitely many. Bounded message size making
DY reachability finite (and decidable) is standard~\citep{RamSu03}.
\textsc{Attacker}'s moves (rule instances with premises in $\know'$, and
acquisitions) are finite at each position; the honest agents' moves are
\emph{determined} by $\vec\ell$ via $\sigma_H$, so they contribute no branching
beyond the protocol's own nondeterminism. The game is a finite-arena,
two-player, zero-sum \emph{quantitative reachability} game with non-negative
integer weights.

\emph{Determinacy and computability.} Such games are determined with optimal
positional strategies, and their values are computable by backward induction over
the budget-augmented arena $(\know',\vec\ell,\mathrm{avail},b)$ with
$b\le B$~\citep{khachiyan2008shortestpath,brihaye2015pareto}: $\mathrm{val}(G)$ is
the least $B$ for which \textsc{Attacker} forces the target within spend $\le B$.
Non-negativity makes the spend monotone, so the budget dimension is acyclic and
the attractor converges in $O(|\mathit{positions}|)$ rounds.

\emph{Qualitative shadow and cheapest-attack identity.} A finite-value winning
play exhibits a DY derivation of a goal message: its computational steps are DY
rule applications and its acquisitions/injections are the intruder's interaction
with the protocol run, i.e.\ a DY attack. Conversely a DY attack against the
protocol is such a derivation interleaved with the honest run; \textsc{Attacker}
replays it. Since the honest coalition's emissions are fixed by the protocol
rather than chosen to deny the intruder, no play can raise the intruder's cost
above that of the cheapest derivation the protocol run makes available; hence
$\mathrm{val}(G(\know_\intruder^0,\varphi)) = \min_\Pi w(\Pi)$ over attacks $\Pi$
realisable against the protocol, which is the deduction cost
$\Delta_w$ of Definition~\ref{def:deltaw} relativised to the messages the run
supplies.
\end{proof}

\section{Rational Security in Weighted ATL}\label{sec:watl}

\subsection{The Logic \watl{}}

We use a weighted fragment of ATL. Formulas are built from atomic propositions,
Boolean connectives, and the cost-bounded strategic modality
\[
\coalition{A}^{\bowtie b}\,\lozenge\,\varphi,\qquad \bowtie\ \in\{<,\le,=,\ge,>\},
\]
read: coalition $A$ has a (uniform) joint strategy to reach a $\varphi$-state
along a path whose accumulated cost stands in relation $\bowtie$ to budget
$b\in\mathbb{N}$. The unweighted modality $\coalition{A}\lozenge\varphi$ is the
special case with no budget constraint. This is a resource-bounded reachability
fragment; it is small, but this is enough for what we aim to express here.

\begin{definition}[Syntax of \watl{}]
Let $\Pi$ be the set of atomic propositions and $\mathit{Ag}$ the set of agents.
State formulas $\varphi$ of \watl{} are generated by the grammar
\[
\varphi \;::=\; p \;\mid\; \neg\varphi \;\mid\; \varphi\wedge\varphi
   \;\mid\; \coalition{A}^{\bowtie b}\,\lozenge\,\varphi ,
\]
where $p\in\Pi$, $A\subseteq \mathit{Ag}$, $b\in\mathbb{N}$, and
$\bowtie\ \in\{<,\le,=,\ge,>\}$. The remaining Boolean connectives are defined as
usual. We write $\coalition{A}\lozenge\varphi$ for
$\coalition{A}^{\ge 0}\lozenge\varphi$, the unconstrained modality, since every
accumulated cost is $\ge 0$.
\end{definition}

\begin{definition}[Semantics of \watl{}]
Let $\mathcal{C}$ be a rational DY CGS and $s\in S$. Satisfaction is defined by
the usual clauses for atoms and Booleans, together with
\[
\mathcal{C},s \models \coalition{A}^{\bowtie b}\,\lozenge\,\varphi
\]
iff there exists a uniform joint strategy $\sigma_A$ for the coalition $A$ such
that, for \emph{every} run $\lambda\in\mathit{out}(s,\sigma_A)$, there is an index
$j\ge0$ with $\mathcal{C},\lambda[j]\models\varphi$ and
\[
\sum_{i<j}\cost\big(\sigma_A(\lambda[{\le}i])\big) \;\bowtie\; b ,
\]
i.e.\ the cost accumulated by $A$ strictly before reaching the first
$\varphi$-state stands in relation $\bowtie$ to the budget $b$. Costs of agents
outside $A$ are not charged to $A$.
\end{definition}

Two features of this fragment deserve emphasis. First, the budget is accumulated
along the path, not per step, so $\coalition{A}^{<b}\lozenge\varphi$ expresses
``$A$ can force $\varphi$ for a total outlay below $b$'', which is precisely the
form a rational-attack witness takes. Second, the only temporal operator is
reachability $\lozenge$; there is no $\mathsf{U}$, no nesting of temporal
operators inside the modality, and no strategy contexts. This is deliberate: the
fragment is exactly as expressive as we need for the ``no positive-utility
attack'' criterion.

\subsection{Rational Security}

Let $\mathit{viol}$ be the proposition marking states that violate the security
property of interest, and let each such state carry reward $\reward$. The
intruder profits iff it can reach $\mathit{viol}$ while spending strictly less
than the reward it collects.

\begin{definition}[Rational security]\label{def:ratsec}
A protocol $P$ with reward function $\reward$ is \emph{rationally secure} iff for
every goal reward value $R$ in the range of $\reward$,
\[
\mathcal{C}, s_0 \;\not\models\; \coalition{\intruder}^{< R}\,\lozenge\,(\mathit{viol}\wedge \reward{=}R).
\]
Equivalently: the intruder has no strategy reaching a violation of value $R$ at
total cost below $R$, i.e., no attack yields positive utility.
\end{definition}

This is the RPD  criterion~\citep{garay2013rational} expressed in the object
logic. It is strictly more permissive than DY security, which is the qualitative
statement $\mathcal{C},s_0\not\models\coalition{\intruder}\lozenge\,\mathit{viol}$.

\begin{proposition}[Refinement]
DY security implies rational security, but not conversely. There exist protocols
$P$ and reward functions such that $P$ is DY-insecure yet rationally secure.
\end{proposition}
\begin{proof}
\emph{DY security implies rational security.} Suppose $P$ is DY-secure, i.e.\
$\mathcal{C},s_0\not\models\coalition{\intruder}\lozenge\,\mathit{viol}$: the
intruder has no strategy whatsoever reaching a violating state. So, clearly, the intruder has
no strategy reaching one within any budget, since a cost-bounded witness is in
particular an unbounded witness. Hence
$\mathcal{C},s_0\not\models\coalition{\intruder}^{<R}\lozenge(\mathit{viol}\wedge
\reward{=}R)$ for every $R$ in the range of $\reward$, which is exactly rational
security (Definition~\ref{def:ratsec}).

\emph{The converse fails.} We exhibit a protocol that is DY-insecure yet
rationally secure. Let $P$ be a protocol whose \emph{only} attack requires the
intruder to inject $n$ distinct forged messages, each of cost $c>0$, before a
violating state is reached; suppose the violation carries reward $R$, and choose
the parameters so that
\[
R \;<\; n\,c .
\]
Formally, take $\mathcal{C}$ with a single violating state $s_{\mathit{viol}}$
reachable from $s_0$ only along paths containing $n$ occurrences of
$\mathsf{inject}$ actions, with $\cost(\mathsf{inject}(\cdot))=c$ and
$\reward(s_{\mathit{viol}})=R$.

The protocol is DY-insecure: the intruder \emph{can} inject the $n$ messages, so
$s_{\mathit{viol}}$ is reachable and
$\mathcal{C},s_0\models\coalition{\intruder}\lozenge\,\mathit{viol}$ holds. But it
is rationally secure: any strategy reaching $s_{\mathit{viol}}$ accumulates cost at
least $nc$, and since $nc > R$, no strategy reaches $s_{\mathit{viol}}$ with
accumulated cost strictly below $R$. Hence the witness formula
$\coalition{\intruder}^{<R}\lozenge(\mathit{viol}\wedge\reward{=}R)$ is
unsatisfiable, so no attack yields positive utility: every attack has utility
$R - nc < 0$, and the intruder's best response is to abstain.

The separation is therefore governed by the single inequality $R \lessgtr nc$,
which is a \emph{computable threshold}: the protocol flips from rationally secure
to rationally insecure exactly when the reward crosses the cheapest attack's cost.
% Section~7 exhibits a natural such $P$, where the threshold is the coercer's
% bribe-to-benefit ratio in ThreeBallot rather than an artificial parameter.
\end{proof}

\subsection{Rational Security via Deduction Value}

We can now state rational security \emph{intrinsically}, in terms of the
deduction layer, and reconcile it with the \watl{} formulation of Section~\ref{sec:watl}.

\begin{definition}[Deduction-rational security] \label{dedS}
Let $\varphi_{\mathit{viol}}$ characterise the messages whose derivation
constitutes a security violation of value $R$. A protocol is
\emph{deduction-rationally secure} iff
\[
\mathrm{val}\big(G(\know_\intruder^{0}, \varphi_{\mathit{viol}})\big) \;\ge\; R,
\]
i.e.\ the cheapest attack the intruder can force costs at least what the violation
is worth.
\end{definition}

\begin{theorem}[Equivalence]
A protocol is deduction-rationally secure in the sense of Definition~\ref{dedS} iff it is rationally secure in the sense
of Definition~\ref{def:ratsec} (no positive-utility \watl{} attack strategy exists).
\end{theorem}
\begin{proof}
We show that a positive-utility \watl{}
attack of value $R$ exists iff $\mathrm{val}(G(\know_\intruder^0,
\varphi_{\mathit{viol}}))<R$. Fix the violation value $R$ and write
$\psi=(\mathit{viol}\wedge\reward{=}R)$.

\emph{($\Rightarrow$) \watl{} witness $\to$ cheap \textsc{Prover} play.} Suppose
$\mathcal{C},s_0\models\coalition{\intruder}^{<R}\lozenge\,\psi$, witnessed by a
uniform strategy $\sigma_\intruder$ whose every outcome run reaches a
$\psi$-state with accumulated cost $<R$. Along any such run, the sequence of
$\mathsf{inject}/\mathsf{compute}$ actions the intruder plays corresponds
step-by-step to applications of weighted DY rules: an $\mathsf{inject}(m)$ is
licensed only when $m\in\overline{\know_\intruder}$, i.e.\ when a derivation of
$m$ exists, and its cost is the same $\cost$ charged in the game. Reading these
rule applications off the run yields a \textsc{Prover} strategy in
$G(\know_\intruder^0,\varphi_{\mathit{viol}})$ whose derivations realise the goal
message; the honest actions the run depends on are exactly the messages
\textsc{Blocker} would have to release, and since the run is an \emph{outcome} of
$\sigma_\intruder$ against the honest protocol, \textsc{Blocker} cannot withhold
them. The total weight of the \textsc{Prover} play equals the accumulated action
cost of the run, which is $<R$. Hence, $\mathrm{val}(G)<R$.

\emph{($\Leftarrow$) cheap \textsc{Prover} play $\to$ \watl{} witness.} Suppose
$\mathrm{val}(G)<R$, witnessed by a \textsc{Prover} strategy of value $<R$ that,
against every \textsc{Blocker}, derives a goal message $m\models
\varphi_{\mathit{viol}}$. We schedule this derivation into a \watl{} strategy.
A weighted DY derivation is a finite proof DAG \cite{Suresh03}; fix any topological order of its
rule applications consistent with premise-before-conclusion. The intruder strategy
$\sigma_\intruder$ plays the corresponding $\mathsf{compute}/\mathsf{inject}$
action at each step, in this order, performing $\mathsf{intercept}$ of an honest
message exactly when the derivation consumes a \textsc{Blocker}-controlled
premise. Two points make this well-defined as a \emph{uniform} strategy. First,
because the topological schedule depends only on the derivation, not on
runtime observations, $\sigma_\intruder$ prescribes the same action on any two
histories in the same $\sim_\intruder$-class, so uniformity holds. Second, the
accumulated action cost telescopes to the weight of the proof DAG, which is
$\mathrm{val}(G)<R$; and the schedule terminates in a state where $m$ has been
derived and injected, i.e.\ a $\psi$-state of value $R$. Thus every outcome run
reaches $\psi$ with cost $<R$, witnessing
$\coalition{\intruder}^{<R}\lozenge\,\psi$.

% Both directions preserve the same numeric quantity (i.e., proof weight equals
% accumulated action cost), so the game value and the least witnessing budget
% coincide, and one is below $R$ iff the other is.
\end{proof}

% This equivalence is a nice theoretical ``gain''. It says the strategic-logic security
% notion and the game-theoretic deduction notion are two views of the same
% quantity (i.e., the value of the attack game) and it lets us compute rational security
% in whichever layer is more convenient: the \watl{} layer for expressing rich
% temporal goals, or the deduction layer for the sharp combinatorial optimisation over
% attack proofs. In particular, the cheapest-derivation computation of
% Proposition~\ref{prop:cheap} becomes a sound and complete decision procedure for rational
% secrecy, giving a deduction-theoretic core to the whole framework.

% \subsection{What Rationality Buys: Structural Consequences}

% Three consequences distinguish the weighted theory from classical DY and justify
% the extra machinery.

% \paragraph{Monotonicity fails, usefully.} Classical DY knowledge is monotone:
% more knowledge never removes an attack. Deduction \emph{value} is monotone in the
% opposite, informative direction (i.e., $\Delta_w$ can only decrease as $\know$
% grows), so the security \emph{margin} $R - \mathrm{val}(G)$ is the right
% sensitivity measure: it quantifies how much cheaper the environment or a
% corruption would have to make an attack before the protocol becomes rationally
% insecure. Classical DY offers no such margin; it reports only presence or absence.

\paragraph{Attack minimality is well-defined.} We note that, in classical DY, ``the attack'' is
any derivation; in our weighted theory, there is a canonical \emph{cheapest} attack,
and Proposition~\ref{prop:cheap} computes it. This gives protocol designers an actionable
target (i.e., raise the cost of the cheapest attack above $R$) rather than the binary
and often unachievable goal of eliminating all attacks.

\paragraph{Composition of costs.} We also note that deduction costs compose additively along proof
structure, so the price of a compound attack decomposes over its sub-derivations.
% This yields a compositional proof principle: to lower-bound
% $\mathrm{val}(G(\know,\varphi))$ it suffices to lower-bound the value of any
% frontier cut through the derivation hypergraph:
We have weighted analogue of the
classical ``the intruder must obtain one of these secrets first'' argument, now
carrying quantitative force.

\section{Decidability and Complexity}\label{sec:dec}

\begin{theorem}[Decidability]
Let $\mathcal{C}$ be a finite rational DY CGS with bounded message depth $\delta$
and integer costs and rewards. Then rational security of $\mathcal{C}$ is
decidable.
\end{theorem}
\begin{proof}
\emph{Finiteness.} By the depth bound $\delta$, the derivable set
$\overline{\know_\intruder(s)}$ contains only terms of depth $\le\delta$ over the
finite atomic-message signature, so it is finite for every $s$. So,
$\Act_\intruder(s)$ (i.e., the intercept/inject/compute actions available from
$\know_\intruder(s)$) is finite, and since $S$, $\mathit{Ag}$ and the honest
action sets are finite, $\mathcal{C}$ is a finite CGS. The reward function has
finite range $\{R_1,\dots,R_k\}\subseteq\mathbb{N}$.

\emph{Reduction to a finite conjunction.} By Definition~\ref{def:ratsec},
rational security is
$\bigwedge_{i=1}^{k}\ \mathcal{C},s_0\not\models
\coalition{\intruder}^{<R_i}\lozenge(\mathit{viol}\wedge\reward{=}R_i)$,
a finite conjunction because $\reward$ has finite range. It therefore suffices to
decide each cost-bounded reachability formula
$\coalition{\intruder}^{<R}\lozenge\,\psi$ with
$\psi=(\mathit{viol}\wedge\reward{=}R)$.

\emph{The budget-augmented arena.} Fix $R$ and build a two-player,
turn-based, zero-sum reachability game $\mathcal{G}_R$ whose protagonist is the
intruder and whose antagonist is the (fixed-strategy) honest environment. Its
positions are pairs $(s,b)\in S\times\{0,1,\dots,R\}$, where $b$ is the
\emph{budget already spent}; the initial position is $(s_0,0)$. From
$(s,b)$ the protagonist chooses an enabled
$a\in\Act_\intruder(s)$ of cost $\cost(a)$; the environment resolves the honest
actions $\vec{a}_H$; and the game moves to $(s',\,b+\cost(a))$ where
$s'=\mathit{tr}(s,(a,\vec{a}_H))$, provided $b+\cost(a)\le R$ (moves that would
exceed the budget are disabled, capping the second component at $R$). The target
set is $T=\{(s,b)\mid s\models\psi,\ b<R\}$: a violation of value $R$ reached
having spent strictly less than $R$.

\emph{Decision by backward induction.} The arena $\mathcal{G}_R$ is finite:
$|S|\cdot(R{+}1)$ positions. Since all weights are non-negative integers and the
budget is monotonically non-decreasing and capped at $R$, no play cycles through
the budget dimension without terminating, so attractor computation converges.
Compute the protagonist's reachability attractor of $T$ by the standard
fixed-point:
$\mathrm{Attr}^0=T$, and $\mathrm{Attr}^{n+1}=\mathrm{Attr}^n\cup\{(s,b)\mid
\exists a\in\Act_\intruder(s)\ \forall \vec{a}_H:\ (s',b{+}\cost(a))\in
\mathrm{Attr}^n\}$, to fixpoint. Then,
$(\mathcal{C},s_0)\models\coalition{\intruder}^{<R}\lozenge\psi$ iff
$(s_0,0)\in\bigcup_n\mathrm{Attr}^n$; equivalently, iff the intruder can force a
value-$R$ violation while keeping total spend below $R$, which is exactly a
positive-utility attack of value $R$. Negating and conjoining over the finitely
many $R_i$ decides rational security.
\end{proof}

\begin{theorem}[Complexity]
Under perfect information for the intruder, rational security is decidable in
time polynomial in $|\mathcal{C}|$ and in the largest reward $R_{\max}$ (pseudo\-/
polynomial in the weights). Under imperfect information (uniform strategies), the
problem is \textsc{PSPACE}-hard.
\end{theorem}
\begin{proof}
\emph{Upper bound (perfect information).} When the intruder has perfect
information (i.e., when $\sim_\intruder$ is the identity) the game $\mathcal{G}_R$ of the
previous proof is a two-player \emph{perfect-information} reachability game on
$|S|\cdot(R{+}1)$ positions. Its attractor is computed in time linear in the
number of edges, i.e.\ $O(|S|\cdot|\Act|\cdot R)$ per reward value, and there are
at most $|S|$ distinct reward values, giving
$O(|S|^2\cdot|\Act|\cdot R_{\max})$ overall. This is polynomial in $|\mathcal{C}|$
and in $R_{\max}$, hence \emph{pseudo-polynomial} in the numeric weights (it
depends on the magnitude $R_{\max}$, not merely its bit-length). The dependence on
$R_{\max}$ is unavoidable in general, as the budget dimension must be tracked.

\emph{Lower bound (imperfect information).} Restricting to \emph{uniform}
strategies (Definition~5) subsumes qualitative $\mathsf{ATL}$ model-checking under
imperfect information and memoryless strategies: set every cost to $0$ and every
reward to $1$, so the budget constraint $<R$ becomes vacuous and the formula
$\coalition{\intruder}^{<1}\lozenge\,\mathit{viol}$ collapses to the plain
$\coalition{\intruder}\lozenge\,\mathit{viol}$. Deciding the latter over an iCGS
with uniform strategies is already \textsc{PSPACE}-hard~\citep{jamroga2006complete},
and the reduction is immediate since our CGS is exactly an iCGS whose
indistinguishability is $\sim_\intruder$. Hence, rational security under imperfect
information is \textsc{PSPACE}-hard.

\emph{Regaining tractability.} The blow-up is caused by the uniform-strategy
quantification over an unrestricted indistinguishability relation, not by the
budget dimension (which only adds the polynomial factor $R_{\max}$). Restricting
the \emph{information flow}, so that the intruder's coalition enjoys the
A-cast/DS-knowledge property of \citet{belardinelli2022acast,belardinelli2025infosharing},
under which the information sets do not grow unboundedly along a strategy, keeps
the underlying qualitative model-checking decidable, and the budget augmentation
of the first proof then rides on top without affecting decidability. 
\end{proof}

\raggedbottom 

% Two remarks are to be made here. First, the budget augmentation is what makes the
% quantities tractable: because costs are non-negative integers and the budget is
% bounded by the reward, only finitely many cost levels matter. Second, imperfect
% information is essential for the voting application to follow: i.e., there, the coercer's inability to
% distinguish a real from a faked vote \emph{is} an indistinguishability of states
% for its uniform strategy, so the imperfect-information case is not a
% generalisation we may discard.

\section{Main Use-case: Rational Attacks under Session Uncertainty}
\label{sec:relay}

The ThreeBallot analysis of the previous section is deliberately
\emph{cryptography-free}: its coercion-resistance is combinatorial, so the DY
deduction layer is trivial and all the work is done by indistinguishability. We
now give a use-case at the opposite pole and beyond it, one in which the DY
\emph{deduction} layer, the \emph{cumulative} structure of the attack, and the
intruder's \emph{imperfect information} are all simultaneously load-bearing. It
abstracts relay attacks and forgery attacks on authenticated,
value-bearing transactions, the setting of deployed contactless-payment relay
attacks~\citep{radu2022emv}. We build it in two steps: a single-session version
that is a one-shot decision, and then the multi-session version that is a 
strategic game, and which is the point of the section.

\paragraph{Our use-case protocol.} This is as follows. A prover $P$ authenticates to a verifier $V$ for a
transaction of value $R$. $V$ sends a fresh challenge $n$; $P$ replies with the
authenticator $a=\{n\}_{k}$ keyed by the long-term secret $k$ it shares with $V$;
$V$ accepts, authorising a payment of $R$, iff it receives $\{n\}_{k}$
(see Figure~\ref{fig:proto}).

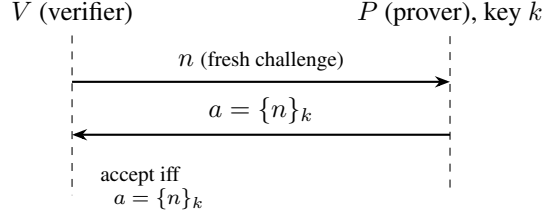
\begin{figure}[t]
\centering
\begin{tikzpicture}[
  >={Stealth[length=2mm]},
  every node/.style={font=\small},
  msg/.style={->, thick},
]
% lifelines
\node (Vt) at (0,0) {$V$ (verifier)};
\node (Pt) at (5,0) {$P$ (prover), key $k$};
\draw[dashed] (Vt) -- ++(0,-2.4);
\draw[dashed] (Pt) -- ++(0,-2.4);
\coordinate (V0) at (0,-0.5);
\coordinate (P0) at (5,-0.5);
% messages
\draw[msg] (0,-0.9) -- node[above]{$n$ \scriptsize(fresh challenge)} (5,-0.9);
\draw[msg] (5,-1.6) -- node[above]{$a=\{n\}_{k}$} (0,-1.6);
\node[align=left] at (0.9,-2.15) {\scriptsize accept iff};
\node[align=left] at (1.15,-2.45) {\scriptsize $a=\{n\}_{k}$};
\end{tikzpicture}
\caption{The base authentication protocol. Acceptance authorises a payment of
value $R$ to whoever presents a valid $a$.}
\label{fig:proto}
\end{figure}

\subsection{Single session: a one-shot decision}

The intruder wants $V$ to accept, i.e.\ to make $V$ receive $a^\star=\{n\}_{k}$
without holding $k$. Two derivations are available, and the weighted DY system
prices both: \\
1. \emph{Relay}: intercept $V$'s challenge, forward it to a genuine
far-away $P$ that answers by protocol \emph{blind to the attack}
(Definition~\ref{def:dedgame}), intercept the reply, inject it to $V$, at
$\cost_{\mathrm{relay}} = 2\,w_{\mathrm{intercept}} + w_{\mathrm{inject}} +
w_{\mathrm{relay}}$. \\
2. \emph{Forge}: corrupt $P$ for $k$, then compute the
authenticator, at
$\cost_{\mathrm{forge}} = w_{\mathrm{corrupt}} + w_{\mathrm{enc}} +
w_{\mathrm{inject}}$; here the composing deduction costs of
Definition~\ref{def:dercost} give the forge cost as literally the cost of the
derivation tree for $\{n\}_k$ (acquire $k$, then one encryption). \\
The cheapest
attack is $\Delta_w(\know,a^\star) =
\min(\cost_{\mathrm{relay}},\cost_{\mathrm{forge}})$, and rational security holds
iff $R \le \Delta_w(\know,a^\star)$.

This is  says something a capability model cannot: a
DY-insecure protocol is \emph{rationally secure} whenever the transaction ceiling
$R$ is below the cheapest attack, which is the formal content of contactless
payment limits. But, as a \emph{single} decision it is not yet strategic: with the
whole protocol run visible, the intruder computes $\Delta_w$ over a bounded set of
derivation trees and compares to $R$. There is no sequencing, no uncertainty, and
the $\sim_\intruder$-machinery is idle, since one representative run stands for
all. The next step supplies exactly what is missing.

\subsection{Two indistinguishable sessions: a strategic game}

$V$ now runs two concurrent sessions whose opening traffic is identical from the
intruder's view, differing only in hidden state. The setup is as follows.
\begin{itemize}
\item \textbf{Session $A$} (the real target): uses key $k_A$, carries the
high-value transaction, reward $R$.
\item \textbf{Session $B$} (a decoy): uses key $k_B$, reward $0$.
\item The two authenticators $\{n\}_{k_A}$ and $\{n\}_{k_B}$ appear on the channel as
\emph{opaque} ciphertexts $c_A,c_B$: the intruder holds neither key, so it cannot
read them or tell which carries the real transaction.
\item The two global states therefore lie in one $\sim_\intruder$-class: they are
\emph{indistinguishable} to the intruder.
\end{itemize}

This is the ``right'' situation the CGS of Section~3. The intruder can
\emph{forward} either ciphertext, one enabled action available in both worlds,
but by the observation/effect split its consequence is resolved by the true
state, causing $V$ to accept only in the matching world (Figure~\ref{fig:attack}).

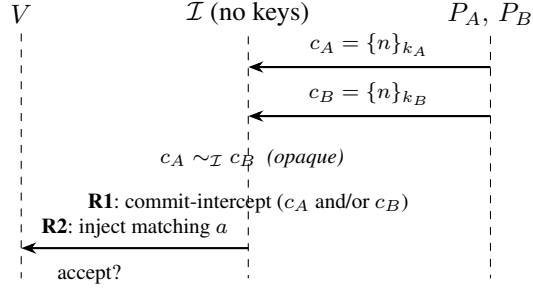
\begin{figure}[t]
\centering
\begin{tikzpicture}[
  >={Stealth[length=2mm]},
  every node/.style={font=\small},
  msg/.style={->, thick},
  ig/.style={font=\small\itshape},
]
\node (V) at (0,0) {$V$};
\node (I) at (3,0) {$\intruder$ (no keys)};
\node (P) at (6.2,0) {$P_A,\,P_B$};
\draw[dashed] (V) -- ++(0,-3.5);
\draw[dashed] (I) -- ++(0,-3.5);
\draw[dashed] (P) -- ++(0,-3.5);
% opaque ciphertexts appear
\draw[msg] (6.2,-0.7) -- node[above,font=\scriptsize]{$c_A=\{n\}_{k_A}$} (3,-0.7);
\draw[msg] (6.2,-1.35) -- node[above,font=\scriptsize]{$c_B=\{n\}_{k_B}$} (3,-1.35);
\node[ig, align=center] at (3.05,-1.9) {\scriptsize $c_A\sim_\intruder c_B$ \ (opaque)};
% round 1 commit: intercept
\node[align=center,font=\scriptsize] at (3,-2.5)
  {\textbf{R1}: commit-intercept ($c_A$ and/or $c_B$)};
% round 2 inject
\draw[msg] (3,-3.1) -- node[above,font=\scriptsize]{\textbf{R2}: inject matching $a$} (0,-3.1);
\node[align=center,font=\scriptsize] at (0.9,-3.45) {accept?};
\end{tikzpicture}
\caption{Two indistinguishable sessions. The intruder sees two opaque
ciphertexts it cannot tell apart; in round~1 it commits to intercepting one (or
both), and in round~2 injects the authenticator, which $V$ accepts only if it
matches the true session.}
\label{fig:attack}
\end{figure}

\paragraph{The attack is now two interdependent rounds.}
\begin{itemize}
\item \textbf{Round 1 (commit, under uncertainty).} The intruder acquires one or
both ciphertexts, each at cost $w_{\mathrm{intercept}}$. This choice is
\emph{sunk}: it is paid before the true session is known, and by uniformity it
must be \emph{the same} in worlds $A$ and $B$, since they are indistinguishable.
\item \textbf{Round 2 (complete).} The intruder injects the authenticator $V$
expects. It can derive this \emph{only} from the ciphertext of the true session:
with only $c_A$ in its knowledge, $\overline{\know_\intruder}$ contains the
session-$A$ authenticator but \emph{not} the session-$B$ one (which needs $k_B$).
A wrong round-1 commit makes the goal underivable, whatever the intruder spends
now.
\end{itemize}

\paragraph{The strategies, and their guaranteed value.} Write
$w_{\mathrm{i}}=w_{\mathrm{intercept}}$ and $w_{\mathrm{j}}=w_{\mathrm{inject}}$.
\begin{itemize}
\item \textbf{Guess $c_A$} (``intercept $c_A$, then inject''). Cost
$w_{\mathrm{i}}+w_{\mathrm{j}}$ on \emph{both} runs of the class. Reaches the goal
only on the $A$-run (reward $R$); on the $B$-run the acquired $c_A$ is useless
(reward $0$). Guaranteed (worst-case) utility
$\min(R,0)-(w_{\mathrm{i}}+w_{\mathrm{j}}) = -(w_{\mathrm{i}}+w_{\mathrm{j}})<0$.
\item \textbf{Guess $c_B$}: symmetric, and no better.
\item \textbf{Hedge} (``intercept \emph{both} $c_A,c_B$, then inject the one that
matches''). Acquiring both removes the uncertainty: whichever session is real, the
intruder holds its ciphertext and can complete the derivation, the second move
being a best response to the now-effectively-revealed world. Guarantees the goal
in either world at cost $2w_{\mathrm{i}}+w_{\mathrm{j}}$, for guaranteed utility
$R-(2w_{\mathrm{i}}+w_{\mathrm{j}})$.
\end{itemize}
The single guesses lose against the worst case; the hedge is the intruder's best
guaranteed play.

\paragraph{The threshold, and why imperfect information changes it.} Rational
security holds iff \emph{no} uniform strategy guarantees positive utility. Since
the hedge is the best the intruder has, the protocol is rationally secure iff
\[
\boxed{\,R \;\le\; 2\,w_{\mathrm{intercept}} + w_{\mathrm{inject}}\,}
\qquad\text{(the \emph{hedge cost}).}
\]
This is strictly higher than the single-session \emph{lucky-guess cost}
$w_{\mathrm{intercept}}+w_{\mathrm{inject}}$, and the gap is the entire content of
the imperfect-information analysis. 

A designer who models the intruder as
\emph{omniscient}, able to tell $A$ from $B$, prices the attack at the lucky-guess
cost and \emph{rejects} the protocol for $R$ in the intermediate range; the
correct rational-DY analysis, accounting for the intruder's uncertainty, prices it
at the strictly higher hedge cost and \emph{certifies} the same protocol as
rationally secure. The intruder's inability to distinguish the sessions is not a
modelling convenience but the source of a quantitatively higher security
guarantee.

\paragraph{Why all use-case ingredients are important.}
\begin{itemize}
\item \emph{Cumulative:} the round-1 intercept is sunk before round~2, and the two
moves are interdependent, the second constrained by the first.
\item \emph{DY-symbolic:} which authenticator the intruder can derive is governed
by the closure of what it committed to acquiring, so a wrong commit renders the
goal underivable.
\item \emph{Imperfect-information:} the uniform round-1 choice must span an
indistinguishable class on which reward varies while cost does not, which is what
turns the one-shot decision into a game with winning and losing strategies and
makes the \emph{hedge}, not the guess, the price of rational security.
\end{itemize}

\section{Use-case: Rational Coercion in ThreeBallot}
\label{sec:threeballot}

Where the use-case of Section~\ref{sec:relay} exercises the deduction layer and
the intruder's strategising under session uncertainty, ThreeBallot sits at the
opposite pole: it uses \emph{no cryptography}, so the DY deduction closure does
nothing,
and the entire security argument rests on indistinguishability and the cost of
communication. Presenting both shows the framework spans the two extremes, the
purely deductive and the purely informational, within one cost-and-incentive
account.

We instantiate the \emph{ThreeBallot}~\citep{rivest2006threeballot} in our framework. It was
modelled in multi-agent-systems literature
\citep{belardinelli2021bisimulations}, so comparing that with our approach is meaningful.

\paragraph{The ThreeBallot Protocol.}
Each voter receives a \emph{multi-ballot} made of three identical ribbons, each
listing all candidates and each bearing a distinct, meaningless serial ID at the
foot. To \emph{vote for} a candidate the voter marks that candidate's row on
\emph{exactly two} of the three ribbons; for \emph{every other} candidate the
voter marks the row on \emph{exactly one} ribbon. (Any other pattern is an
invalid ballot.) The voter separates the three ribbons and casts all three into a
ballot box that scrambles their order, destroying the link between ribbons of one
multi-ballot. After the poll, every cast ribbon is scanned and published on a
public bulletin board. The tally for a candidate is the total marks on that
candidate's row across all ribbons, minus the number of voters, since each voter
contributes exactly one ``baseline'' mark per candidate and one extra mark to
their chosen candidate.

Two features give ThreeBallot its coercion story. First, the voter keeps a
\emph{receipt}: a copy of one of her three ribbons, chosen by her. She can later
check that this exact ribbon appears on the bulletin board, giving individual
verifiability. Secondlt, a single ribbon, in isolation, is
consistent with many different votes: because the ``vote for'' pattern is two
marks out of three and ``not vote'' is one out of three, any one ribbon could
belong to a multi-ballot expressing almost any preference. Hence, the receipt does
\emph{not} prove how the voter voted. This is exactly what a rational coercer must
contend with.

\paragraph{Setting: A Rational Coercer.} A coercer $\intruder$ approaches a voter $V$ and demands a vote for candidate $x$,
offering a bribe $B$ and, to enforce it, demanding that $V$ hand over her receipt
ribbon as ``proof.'' Following the standard threat
model~\citep{belardinelli2021bisimulations}, we take $\intruder$ to be an agent
who also observes the public bulletin board; $V$ fully cooperates in appearance.
The coercer's actions carry costs: $\mathsf{approach}$ (cost $c_a$: issue the
demand and collect the receipt), $\mathsf{pay}(B)$ (cost $B$: hand over the
bribe), $\mathsf{verify}$ (cost $c_v$: cross-check the surrendered ribbon against
the board). The reward $R$ is the value to $\intruder$ of one genuine vote for
$x$.

The voter's private choice is whether to \emph{comply} (i.e., mark $x$ twice and hand
over a ribbon consistent with the demand) or to \emph{resist}: vote for her true
preference $y\ne x$, yet still produce a surrendered ribbon that is \emph{equally}
consistent with an $x$-vote. Because a lone ribbon under-determines the vote, such
a ribbon always exists; the coercion-resistance of ThreeBallot is precisely the
statement that the coercer's observation (i.e., the surrendered ribbon plus the public
board) is identical in the comply and resist cases. In our model this is the
indistinguishability $h_{\mathit{comply}}\sim_\intruder h_{\mathit{resist}}$ of
the two histories.

\paragraph{Modelling and Analysis.}

We model $\intruder$, $V$, and the board/environment as a rational DY CGS. The
knowledge component $\know_\intruder(s)$ records the surrendered ribbon and the
published board; by the combinatorial argument above, the comply- and
resist-histories fall in the same $\sim_\intruder$-class. The violating
proposition $\mathit{viol}$ holds when a genuine $x$-vote by $V$ is counted
\emph{and} $\intruder$ has paid (i.e., the coercer got what it paid for) carrying
reward $R$. We ask whether the coercer has a profitable attack:
\[
\mathcal{C},s_0 \models
\coalition{\intruder}^{<R}\,\lozenge\,(\mathit{viol}\wedge\reward{=}R).
\]

Because $h_{\mathit{comply}}\sim_\intruder h_{\mathit{resist}}$, any
\emph{uniform} coercer strategy (Definition~5) must prescribe the same action on
both: $\mathsf{verify}$ returns the same observation (i.e., a well-formed ribbon
consistent with an $x$-vote appears on the board) whether $V$ complied or
resisted. So no uniform strategy can condition payment on genuine compliance, and
$\intruder$ faces a decision under its own indistinguishability:
\begin{itemize}
\item \textbf{Pay.} Expected utility $R\cdot p - (c_a + B + c_v)$, where $p$ is
the probability $V$ actually complied. A rational $V$ with true preference
$y\ne x$ and any private benefit $\beta>0$ from voting sincerely resists whenever
$\beta>0$ (i.e., resisting is costless and, by coercion-resistance, undetectable), so
$p\to 0$ and the coercer's utility $\to -(c_a+B+c_v)<0$.
\item \textbf{Do not pay.} Utility $0$ (or $-c_a$ if the approach was already
made).
\end{itemize}
Hence, for a rational voter the coercer's best response is \emph{not to attack}:
the witness formula is unsatisfiable and ThreeBallot is \emph{rationally secure}.
The qualitative model cannot state this distinction: syntactically a coerced state
\emph{is} reachable (i.e., a complying voter exists) so plain $\mathsf{ATL}$ reports
$\coalition{\intruder}\lozenge\,\mathit{viol}$ as \emph{true}; yet no
\emph{rational} coercer brings it about. (This is  the qualitative
coercion-resistance property verified for ThreeBallot
in~\citep{belardinelli2021bisimulations}, now refined into an incentive
statement.)

\paragraph{The Threshold.}

The separation becomes quantitative once we relax perfect coercion-resistance: that is,
suppose that the receipt mechanism leaks compliance with small probability $\epsilon$:
e.g.\ a poorly randomised ribbon-selection or a board-correlation attack lets
$\mathsf{verify}$ distinguish comply from resist with advantage $\epsilon$. Then
$\intruder$ can condition payment on the leak, and its expected utility from
attacking is
\[
\util_\intruder \;=\; \epsilon\,R \;-\; (c_a + \epsilon B + c_v).
\]
Attacking is rational iff $\util_\intruder>0$, i.e.
\[
\boxed{\;\epsilon \;>\; \frac{c_a + c_v}{R - B}\;}
\]
(for $R>B$). Below this leakage threshold no rational coercer attacks; above it,
coercion pays off. Classical coercion-resistance is the single point $\epsilon=0$; our
analysis shows that what actually protects voters is not the binary property but
the \emph{margin} between the leakage the protocol permits and the coercer's
cost-to-net-benefit ratio. 
% This is the guarantee a deployment engineer needs and
% the qualitative theory, including the qualitative ThreeBallot analysis, cannot
% provide.

\section{Conclusions \& Future Work}
 We have shown that the DY attacker can be made \emph{rational} inside a strategic
 logic, and that the resulting notion of rational
 security is a computable, strictly finer guarantee than DY security. 
Several directions
follow. For instance, our decidability rests
on bounded message depth; lifting this to unbounded spaces is desirable. Or, the
threshold analysis above is naturally probabilistic, suggesting a stochastic
 \watl{} in which the coercer maximises expected payoff. 
%
%\emph{Coalitions of
% rational insiders}: the ATL coalition modality $\coalition{A}$ extends directly
% to a corrupt-insider-plus-intruder coalition, where the cost of corruption enters
% the utility: the natural next case, and the point at which the strategic-logic
% formulation earns its keep over a single-attacker game.


\begin{thebibliography}{11}
% \providecommand{\natexlab}[1]{#1}
% \providecommand{\url}[1]{\texttt{#1}}
% \expandafter\ifx\csname urlstyle\endcsname\relax
%   \providecommand{\doi}[1]{doi: #1}\else
%   \providecommand{\doi}{doi: \begingroup \urlstyle{rm}\Url}\fi

\bibitem[Alur et~al.(2002)Alur, Henzinger, and Kupferman]{alur2002alternating}
Rajeev Alur, Thomas~A Henzinger, and Orna Kupferman.
\newblock Alternating-time temporal logic.
\newblock \emph{Journal of the ACM}, 49\penalty0 (5):\penalty0 672--713, 2002.

\bibitem[Ballot et~al.(2024)Ballot, Malvone, Leneutre, and
  Laarouchi]{ballot2024capacity}
Gabriel Ballot, Vadim Malvone, Jean Leneutre, and Youssef Laarouchi.
\newblock Strategic reasoning under capacity-constrained agents.
\newblock In \emph{Proc. 23rd Int. Conf. on Autonomous Agents and Multiagent
  Systems (AAMAS)}, pages 123--131, 2024.

\bibitem[Ballot(2025)]{ballot2025thesis}
Gabriel Ballot.
\newblock \emph{Strategic Reasoning for Cyber-Security}.
\newblock PhD thesis, Institut Polytechnique de Paris, 2025.

\bibitem[Belardinelli et~al.(2025)Belardinelli, Boureanu, Dima, and
  Malvone]{belardinelli2025infosharing}
Francesco Belardinelli, Ioana Boureanu, Catalin Dima, and Vadim Malvone.
\newblock Model-checking strategic abilities in information-sharing systems.
\newblock \emph{ACM Transactions on Computational Logic}, 26\penalty0
  (1):\penalty0 1--45, 2025.

\bibitem[Belardinelli et~al.(2022)Belardinelli, Boureanu, Dima, and
  Malvone]{belardinelli2022acast}
Francesco Belardinelli, Ioana Boureanu, Catalin Dima, and Vadim Malvone.
\newblock Model checking strategic abilities in information-sharing systems.
\newblock \emph{arXiv preprint arXiv:2204.08896}, 2022.

\bibitem[Belardinelli et~al.(2021)Belardinelli, Condurache, Dima, Jamroga, and
  Knapik]{belardinelli2021bisimulations}
Francesco Belardinelli, Rodica Condurache, Catalin Dima, Wojciech Jamroga, and
  Michal Knapik.
\newblock Bisimulations for verifying strategic abilities with an application to
  the {ThreeBallot} voting protocol.
\newblock \emph{Information and Computation}, 276:\penalty0 104552, 2021.

\bibitem[Bella and Bistarelli(2004)]{bella2001rational}
Giampaolo Bella and Stefano Bistarelli.
\newblock Soft constraint programming to analysing security protocols.
\newblock In \emph{Theory and Practice of Logic Programming}. 2004.

\bibitem[Brihaye et~al.(2015)Brihaye, Geeraerts, Haddad, Monmege, Sassolas, and
  Wiedemann]{brihaye2015pareto}
Thomas Brihaye, Gilles Geeraerts, Axel Haddad, Benjamin Monmege, Mathieu
  Sassolas, and Bertrand Wiedemann.
\newblock Quantitative reachability games: Optimal strategies and value
  computation.
\newblock In \emph{Reachability Problems}, Lecture Notes in Computer Science.
  Springer, 2015.

\bibitem[Clarkson et~al.(2008)Clarkson, Chong, and Myers]{clarkson2008civitas}
Michael~R Clarkson, Stephen Chong, and Andrew~C Myers.
\newblock Civitas: Toward a secure voting system.
\newblock In \emph{IEEE S\&P}, pages 354--368, 2008.

\bibitem[Delaune et~al.(2009)Delaune, Kremer, and Ryan]{delaune2010verifying}
St{\'e}phanie Delaune, Steve Kremer, and Mark Ryan.
\newblock Verifying privacy-type properties of electronic voting protocols.
\newblock 2009.

\bibitem[Dima and Tiplea(2011)]{dima2011undecidable}
Catalin Dima and Ferucio~Laurentiu Tiplea.
\newblock Model-checking {ATL} under imperfect information and perfect recall
  semantics is undecidable.
\newblock \emph{arXiv preprint arXiv:1102.4225}, 2011.

\bibitem[Dolev and Yao(1983)]{dolev1983security}
Danny Dolev and Andrew Yao.
\newblock On the security of public key protocols.
\newblock \emph{IEEE Transactions on Information Theory}, 29\penalty0
  (2):\penalty0 198--208, 1983.

\bibitem[Gallo et~al.(1993)Gallo, Longo, Pallottino, and
  Nguyen]{gallo1993hypergraphs}
Giorgio Gallo, Giustino Longo, Stefano Pallottino, and Sang Nguyen.
\newblock Directed hypergraphs and applications.
\newblock \emph{Discrete Applied Mathematics}, 42\penalty0 (2--3):\penalty0
  177--201, 1993.

\bibitem[Garay et~al.(2013)Garay, Katz, Maurer, Tackmann, and
  Zikas]{garay2013rational}
Juan Garay, Jonathan Katz, Ueli Maurer, Bj{\"o}rn Tackmann, and Vassilis Zikas.
\newblock Rational protocol design: Cryptography against incentive-driven
  adversaries.
\newblock In \emph{FOCS}, pages 648--657, 2013.

\bibitem[Gutierrez et~al.(2017)Gutierrez, Harrenstein, and
  Wooldridge]{gutierrez2017rational}
Julian Gutierrez, Paul Harrenstein, and Michael Wooldridge.
\newblock From model checking to equilibrium checking: Reactive modules for
  rational verification.
\newblock \emph{Artificial Intelligence}, 248:\penalty0 123--157, 2017.

\bibitem[Halpern and Teague(2004)]{halpern2004rational}
Joseph Halpern and Vanessa Teague.
\newblock Rational secret sharing and multiparty computation.
\newblock In \emph{STOC}, pages 623--632, 2004.

\bibitem[Jamroga and Dix(2006)]{jamroga2006complete}
Wojciech Jamroga and J{\"u}rgen Dix.
\newblock Model checking abilities under incomplete information is indeed
  $\Delta_2^{\mathrm{P}}$-complete.
\newblock In \emph{Proc. EUMAS'06}, 2006.

\bibitem[Radu et~al.(2022)Radu, Chothia, Newton, Boureanu, and
  Chen]{radu2022emv}
Andreea-Ina Radu, Tom Chothia, Christopher J.P. Newton, Ioana Boureanu, and
  Liqun Chen.
\newblock Practical {EMV} relay protection.
\newblock In \emph{IEEE Symposium on Security and Privacy (S\&P)}, pages
  1737--1756, 2022.

\bibitem[Rivest(2006)]{rivest2006threeballot}
Ronald~L. Rivest.
\newblock The {ThreeBallot} voting system.
\newblock Unpublished manuscript, MIT, 2006.

\bibitem[Suresh(2003)]{Suresh03}
S.~P. Suresh.
\newblock {\em Foundations of Security Protocol Analysis}.
\newblock PhD thesis, Institute of Mathematical Sciences, Chennai, India, 2003.

\bibitem[Juels et~al.(2005)Juels, Catalano, and Jakobsson]{juels2005coercion}
Ari Juels, Dario Catalano, and Markus Jakobsson.
\newblock Coercion-resistant electronic elections.
\newblock In \emph{WPES}, pages 61--70, 2005.

\bibitem[Khachiyan et~al.(2008)Khachiyan, Boros, Borys, Elbassioni, Gurvich, and
  Makino]{khachiyan2008shortestpath}
Leonid Khachiyan, Endre Boros, Konrad Borys, Khaled Elbassioni, Vladimir
  Gurvich, and Kazuhisa Makino.
\newblock On short paths interdiction problems: Total and node-wise limited
  interdiction.
\newblock \emph{Theory of Computing Systems}, 43\penalty0 (2):\penalty0
  204--233, 2008.

\bibitem[Kremer and Raskin(2002)]{kremer2003game}
Steve Kremer and Jean-Fran{\c{c}}ois Raskin.
\newblock Game analysis of abuse-free contract signing.
\newblock In \emph{CSFW}, 2002.

\bibitem[Zeeman(1980)]{gt}
E.~C. Zeeman.
\newblock Population dynamics from game theory.
\newblock In {\em Global Theory of Dynamical Systems}, volume 819 of {\em Lecture Notes in Mathematics}, pages 471--497. Springer, 1980.

\bibitem[Ramanujam and Suresh(2003)]{RamSu03}
R.~Ramanujam and S.~Suresh.
\newblock A decidable subclass of unbounded security protocols.
\newblock In {\em Proc. of the IFIP WG 1.7 and ACM SIGPLAN Workshop on Issues in the Theory of Security (WITS'03)}, pages 11--20. IOS Press, 2003.


\end{thebibliography}
\end{document}